\documentclass[aps,pra,10pt,reprint,twocolumn,groupedaddress,floatfix,%
  longbibliography,notitlepage,showpacs]{revtex4-1}

\usepackage{hyperref}
\usepackage{graphicx}
\usepackage{bm}
\usepackage{amsthm}
\usepackage{amsmath}
\usepackage{thmtools,thm-restate}
\usepackage{amssymb}
\usepackage{mathrsfs}
\usepackage{color}

\newtheorem{theorem}{Theorem}
\newtheorem{lemma}[theorem]{Lemma}

\def\rank{\mathop{\rm rank}}

\def\wgt{\mathop{\rm wgt}}

\begin{document}

\title{Numerical and analytical bounds on threshold error rates for
  hypergraph-product codes}

\author{Alexey A. Kovalev} \affiliation{Department of Physics and
  Astronomy and Nebraska Center for Materials and Nanoscience,
  University of Nebraska, Lincoln, Nebraska 68588, USA}

\author{Sanjay Prabhakar} \affiliation{Department of Physics and
  Astronomy and Nebraska Center for Materials and Nanoscience,
  University of Nebraska, Lincoln, Nebraska 68588, USA}

\author{Ilya Dumer}
\affiliation{Department of Electrical Engineering, University of
  California, Riverside, California 92521, USA} 

\author{Leonid P. Pryadko}
\affiliation{Department of Physics \& Astronomy, University of
  California, Riverside, California 92521, USA} 

\date{\today}

\begin{abstract}
  We study analytically and numerically decoding properties of finite
  rate hypergraph-product quantum LDPC codes obtained from random
  $(3,4)$-regular Gallager codes, with a simple model of independent
  $X$ and $Z$ errors.  Several non-trival lower and upper bounds for
  the decodable region are constructed analytically by analyzing the
  properties of the homological difference, equal minus the logarithm
  of the maximum-likelihood decoding probability for a given syndrome.
  Numerical results include an upper bound for the decodable region
  from specific heat calculations in associated Ising models, and a
  minimum weight decoding threshold of approximately $7\%$.
\end{abstract}


\maketitle

\section{Introduction}
\label{sec:intro}

Coherence protection is one of the key technologies required for
scalable quantum computation.  Quantum error correction is one such
technique.  It enables scalable quantum computation with a
polylogarithmic overhead per logical qubit as long as the accuracy of
elementary gates and measurements exceeds certain
threshold\cite{Shor-FT-1996,Knill-error-bound,Kitaev-QC-1997,%
  Aharonov-BenOr-1997}.  For a given family of quantum error
correcting codes (QECCs), the actual threshold value depends, e.g., on
hardware architecture, implementation of the elementary gates, and on
the algorithm used for syndrome-based decoding.  While these details
are ultimately very important, for the purposes of comparing different
families of QECCs, one is also interested in the threshold(s) computed
in simple ``channel'' models where errors on different qubits are
assumed independent and identically distributed (i.i.d.), with the
assumption of perfect syndrome measurement.  The resulting threshold
is a single number which depends on the chosen algorithm for
syndrome-based decoding.  Among any decoders, the threshold is maximal
for the (exponentially expensive) maximum-likelihood (ML) decoder.

While the original version of the threshold theorem was based on
concatenated
codes\cite{Shor-FT-1996,Knill-error-bound,Fowler-QEC-2004,%
  Aliferis-Gottesman-Preskill-2006}, much better thresholds are
obtained with topological
surface\cite{Dennis-Kitaev-Landahl-Preskill-2002,%
  Wang-Fowler-Stephens-Hollenberg-2010,Wang-Fowler-Hollenberg-2011}
and related topological color
codes\cite{Wang-Fowler-Hill-Hollenberg-2010,Landahl-2011}.  Even
though their thresholds are higher, all surface codes, and, generally,
all codes local in $D$ dimensions, are necessarily zero-rate
codes\cite{Bravyi-Terhal-2009,Bravyi-Poulin-Terhal-2010}.  Just like
codes obtained by repeated concatenation, such local codes also
require an overhead (per logical qubit) that is increasing with the
length of the code, so that the overall overhead must increase as the
size of the computation grows.

Scalable quantum computation with a \emph{constant} overhead can be
potentially achieved\cite{Gottesman-overhead-2014} using more general
quantum LDPC (low density parity-check) codes.  These are stabilizer
codes, with the property that each stabilizer generator involves a
bounded number of qubits.  Here, a non-zero fault-tolerant threshold
is guaranteed if the distance scales logarithmically or faster with
the block
length\cite{Kovalev-Pryadko-FT-2013,Dumer-Kovalev-Pryadko-bnd-2015}.
At the same time, such codes may achieve a finite rate only if their
generators remain non-local whenever the qubits are laid out in a
Euclidean space of finite
dimension\cite{Bravyi-Terhal-2009,Bravyi-Poulin-Terhal-2010}.  Several
infinite code families are known to achieve these
requirements\cite{Tillich-Zemor-2009,Delfosse-Zemor-2010,Delfosse-2013,%
  Kovalev-Pryadko-Hyperbicycle-2013,Guth-Lubotzky-2014}.

The ML decoding probability for a quantum LDPC code can be formally
expressed as the average of a ratio of partition functions for two
associated random-bond Ising models (RBIM)
\cite{Dennis-Kitaev-Landahl-Preskill-2002,Bombin-PRA-2010}, computed
on the Nishimori line\cite{Nishimori-1980,Nishimori-book} in the
$(p,T)$ plane, where $p$ is the error probability and $T$ is the
temperature.  A temperature not on the Nishimori line
corresponds to a suboptimal decoder which assumes an incorrect error
probability.  For topological codes local in $D$ dimensions, the
decodable region is a subset of (and possibly coincides with) the
thermodynamical phase of RBIM where certain extended topological
defects have finite tension\cite{Dennis-Kitaev-Landahl-Preskill-2002,%
  Wang-Harrington-Preskill-2003,Katzgraber-Bombin-MartinDelgado-2009,%
  Bombin-PRA-2010,%
  Bombin-PRX-2012,Katzgraber-Andrist-2013,%
  Jouzdani-Novais-Tupitsyn-Mucciolo-2014,Kovalev-Pryadko-SG-2015}.
For finite-rate codes, decodability requires that the average defect
tension be sufficiently large\cite{Kovalev-Pryadko-SG-2015}.

In this work we analyze error-correcting properties of the finite-rate
family of hypergraph-product codes\cite{Tillich-Zemor-2009} based on
random $(3,4)$-regular Gallager ensemble of classical LDPC codes, in
conjunction with the phase diagrams of the two mutually dual
associated RBIMs, constructed assuming independent $X$ and $Z$ errors
which happen independently with probability $p$ at each qubit (see
Fig.~\ref{fig1}).  More specifically, we use a large-distance subset
of the random Gallager codes.  Each constructed code is a CSS
code\cite{Calderbank-Shor-1996,Steane-1996} with parameters
$[[n,k,d]]$.  Here the asymptotic rate $R=k/n=1/25$, and the distance
$d$ scales as a square root of the code length $n$.  The first
corresponding Ising model has $n$ interaction terms (``bonds'') and
$r\le 12 n/25$ spins; each bond is a product of $3$ or $4$ spin
variables, and each spin participates in up to seven bonds.  These
numbers are higher for the second (dual) model which includes a
summation over additional spin variables corresponding to the
codewords.

\begin{figure}[htbp]
\includegraphics[width=\columnwidth ]{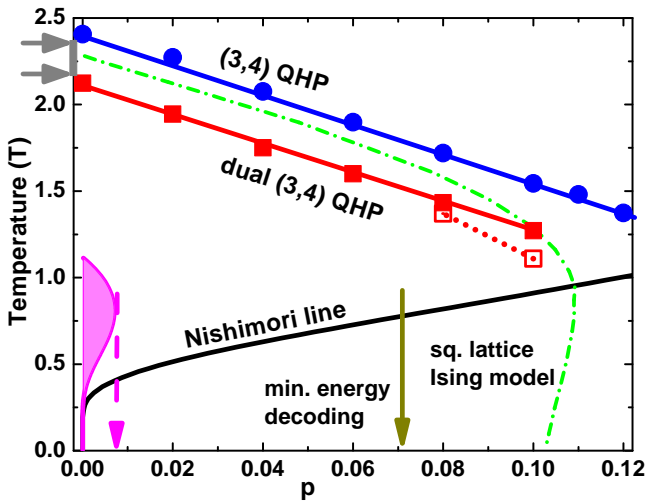}
\caption{(Color online) The $(p,T)$ phase diagram of the two
  random-bond Ising models associated with QHP codes from $(3,4)$
  Gallager ensemble, labeled ``(3,4) QHP'' and ``dual (3,4) QHP'',
  where the latter model includes a summation over additional spin
  variables corresponding to the codewords, see
  Eq.~(\ref{eq:succ-decoding}).  Here $p$ is the error probability and
  $T$ the dimensionless temperature.  Positions of the specific heat
  maxima at different $p$, extrapolated to infinite system size, are
  shown with solid blue circles and solid red boxes with ad hoc linear
  fits [red open boxes connected with a dotted line correspond to
  parabolic extrapolation in Fig.~\ref{fig5}(b)].  The two transition
  lines intersect the $p=0$ axis in approximately mutually dual
  temperatures, and are located, respectively, above and below the
  critical temperature for the square-latice random-bond Ising model
  (dashed-dotted green line; data from
  Ref.~\onlinecite{Thomas-Katzgraber-2011}).  The ML decoding problem
  corresponds to the points on the Nishimori line shown with a solid
  black line; a temperature not on the Nishimori line corresponds to a suboptimal decoder which assumes an incorrect error probability.  The magenta-shaded region shows the lower
  bound for the decodable region from Theorem
  \ref{th:lts-convergence-codes}; the right-most point of this region
  coincides with the lower bound obtained by analyzing
  min\-i\-mum-energy decoder in
  Ref.~\onlinecite{Dumer-Kovalev-Pryadko-bnd-2015} (dashed magenta
  downward arrow).  The forest-green solid vertical arrow shows our
  numerical estimate for the minimum weight decoding threshold, see
  Sec.~\ref{sec:mw-decoding}.  The temperatures $T_\mathrm{max}$ and
  its dual, $T_\mathrm{max}^*$ from Theorem \ref{th:upper-homo-bound}
  are marked with a pair of horizontal arrows separated by a gray bar
  on the vertical axis; the lower arrow corresponds to the analytical
  upper temperature bound of the decodable region.  More accurate
  upper bound from the present data is given by the ``dual (3,4) QHP''
  line.}
\label{fig1}
\end{figure}

Analytically, we construct a lower bound for the decodable region, and
study the relation between the decodability and thermodynamical phases
of the corresponding Ising model.  In particular, this produces a
non-trivial upper temperature bound for the decodable region.  
Numerically, we use  Metropolis updates in canonical ensemble
simulations and in feedback-optimized parallel tempering Monte Carlo
method to compute average specific heat
as a function of temperature and the flipped bond probability $p$;
extrapolation to infinite code distance gives the transition
temperatures in the two
models.  
We give an argument that it is the transition temperature of the dual
model that gives a more accurate
estimate of the upper temperature bound of the decodable region.  We
also use a decoder approximating the minimum-energy decoder to obtain
a lower bound for the ML decoding threshold, with the result
$p_c\ge p_\mathrm{minE}= 7.0\%$.

The rest of the paper is organized as follows.  In
Sec.~\ref{sec:background}, we give a brief overview of classical and
quantum error correcting codes, the Ising models related to ML
decoding, and quantum hypergraph-product codes.  We give the
analytical bounds for the decodable region in Sec.~\ref{sec:qec}, with
the proofs given in the Appendices.  The numerical techniques and the
corresponding results are presented in Sec.\ \ref{sec:numerics}.  We
summarize our results and give some concluding remarks in
Sec.~\ref{sec:conclusion}.

\section{Background}
\label{sec:background}
\subsection{Classical and quantum error correcting codes}

A classical binary linear code $\mathcal{C}$ with parameters $[n,k,d]$
is a $k$-dimensional subspace of the vector space $\mathbb{F}_2^n$ of
all binary strings of length $n$.  Code distance $d$ is the minimal
weight (number of non-zero elements) of a non-zero string in the code.
A code $\mathcal{C}\equiv \mathcal{C}_G$ can be specified in terms of
the generator matrix $G$ whose rows are the basis vectors of the
code.  All vectors orthogonal to the rows of $G$ form the dual code
$\mathcal{C}_G^\perp=\{\mathbf{c}\in\mathbb{F}_2^n|G\mathbf{c}^T=\mathbf{0}\}$.
The generator matrix $P$ of the dual code,
$\mathcal{C}_G^\perp\equiv \mathcal{C}_P$, 
\begin{equation}
  \label{eq:dual-matrix}
  GP^T=0,\quad \rank G+\rank P=n,  
\end{equation}
is also called dual of $G$, $P=G^*$.  It is the parity check matrix of
the original code, $\mathcal{C}_G=\mathcal{C}_P^\perp$.

A quantum $[[n,k,d]]$ stabilizer code is a $2^{k}$-dimensional
subspace of the $n$-qubit Hilbert space $\mathbb{H}_{2}^{\otimes n}$,
a common $+1$ eigenspace of all operators in an Abelian
\emph{stabilizer group} $\mathscr{S}\subset\mathscr{P}_{n}$,
$-\openone\not\in\mathscr{S}$, where the $n$-qubit Pauli group
$\mathscr{P}_{n}$ is generated by tensor products of the $X$ and $Z$
single-qubit Pauli operators.  The stabilizer is typically specified
in terms of its generators,
$\mathscr{S}=\left\langle S_{1},\ldots,S_{n-k}\right\rangle $. The
weight of a Pauli operator is the number of qubits that it
affects. The distance $d$ of a quantum code is the minimum weight of
an operator $U$ which commutes with all operators from the stabilizer
$\mathscr{S}$, but is not a part of the stabilizer,
$U\not\in\mathscr{S}$. Such operators correspond to the logical qubits
and are called logical operators. A Pauli operator
$U\equiv i^{m}X^{\mathbf{v}}Z^{\mathbf{u}}$, where
$\mathbf{v},\mathbf{u}\in\{0,1\}^{\otimes n}$ and
$X^{\mathbf{v}}=X_{1}^{v_{1}}X_{2}^{v_{2}}\ldots X_{n}^{v_{n}}$,
$Z^{\mathbf{u}}=Z_{1}^{u_{1}}Z_{2}^{u_{2}}\ldots Z_{n}^{u_{n}}$, can
be mapped, up to a phase, to a binary vector
$\mathbf{e}=(\mathbf{u},\mathbf{v})$.  With this map, generators of
the stabilizer group are mapped to rows of a generator matrix
$\mathcal{G}=(G^{x},G^{z})$ forming a binary classical linear
code\cite{Calderbank-1997}.  We will also consider the matrix
$\mathcal{L}$ obtained in a similar fashion from independent logical
operators.

For a more narrow set of CSS
codes\cite{Calderbank-Shor-1996,Steane-1996} the stabilizer generators
can be chosen as products of either only $X$ or only $Z$ Pauli
operators.  The corresponding generator matrix is a direct sum,
$\mathcal{G}=G_{x}\oplus G_{z}$, where rows of the matrices
$G\equiv G_x$ and $H\equiv G_z$ are orthogonal, $G_x G_z^T=0$.  For
any CSS code, independent logical operators can also be chosen as
products of only $X$ or only $Z$ Pauli operators, which gives
$\mathcal{L}=L_{x}\oplus L_{z}$.  Rows of the matrix $L_x$ are
orthogonal to rows of $G_z$, $G_z L_x^T=0$, and they are linearly
independent from rows of $G_x$.  Similarly, rows of the matrix $L_z$
are orthogonal to rows of $G_x$, and they are linearly independent
from rows of $G_z$.  For a CSS code of block length $n$, these
matrices have $n$ columns, and the number of encoded qubits is
\begin{equation}
  \label{eq:k-defined}
  k=\rank L_x=\rank L_z=n-\rank G_x-\rank G_z.  
\end{equation}
Rows of $L_x$ and $L_z$, respectively, have weights that are bounded
from below in terms of the corresponding CSS distances,
\begin{equation}
d_x\equiv \min_{\mathbf{c}\in \mathcal{C}_{G_z}^\perp\setminus
  \mathcal{C}_{G_x}}
\wgt(\mathbf{c}),\quad
d_z\equiv \min_{\mathbf{b}\in \mathcal{C}_{G_x}^\perp\setminus
  \mathcal{C}_{G_z}}
\wgt(\mathbf{b}).
\label{eq:CSS-dist}
\end{equation}
The code distance is just $d\equiv \min (d_{G_x}, d_{G_z})$.

In what follows, we concentrate on CSS codes.  It will be convenient
to assume that matrices $L_x$ and $L_z$ have full row rank (each has
exactly $k$ rows), and specifically define the form of the dual matrix
$G_x^*$ [see Eq.~(\ref{eq:dual-matrix})] as a combination of rows of
matrices $G_z$ and $L_z$, and, similarly, the dual matrix $G_z^*$ as a
combination of rows of $G_x$ and $L_x$.  Also, to simplify the
notations, it will be convenient to drop the indices $x$ and $z$ and
use the matrices $G\equiv G_x$ and $H\equiv G_z$.  The corresponding
CSS distances (\ref{eq:CSS-dist}) will be denoted as $d_G\equiv d_x$
and $d_H\equiv d_z$.

\subsection{Maximum likelihood decoding and random-bond Ising model}
\label{sec:ml-decoding}

Consider a CSS code with generator matrices $G\equiv G_x$ and
$H\equiv G_z$, and an error model where bit-flip and phase-flip
errors happen independently with the same probability $p$.  In such a
case, decoding of $X$ and $Z$ errors can be done separately.  In the
following, we only consider $X$ errors.

Generally, an $X$ error can be described by a length-$n$ binary vector
$\mathbf{e}$; errors obtained by adding linear combinations of rows of
$G$ are mutually \emph{degenerate} (equivalent), they act identically
on the code.  In the absence of measurement errors, one needs to
figure out the degeneracy class of the error from the measured
\emph{syndrome} vector, $\mathbf{s}^T=H\mathbf{e}^T$.  While it is
easy to come up with a vector $\mathbf{e}_0$ that satisfies these
equations, so do $2^{k}-1$ vectors $\mathbf{e}_0+\mathbf{c}$ obtained
by adding inequivalent codewords
$\mathbf{c}\in\mathcal{C}_{H}^\perp\setminus \mathcal{C}_G$.  For the
maximum-likelihood (ML) decoding, one compares the probabilities of
errors in different degeneracy sectors (inequivalent $\mathbf{c}$),
and chooses the most likely.

The probability of an error degenerate with $\mathbf{e}$ is obtained
as a sum of probabilities of errors $\mathbf{e}+\boldsymbol{\alpha}G$
for different binary $\boldsymbol{\alpha}$.  Such a sum can be readily
seen\cite{Dennis-Kitaev-Landahl-Preskill-2002} to be proportional to
the partition function of RBIM in Wegner's form\cite{Wegner-1971},
\begin{equation}
  \label{eq:Z}
  Z_\mathbf{e}(G;K_p)=\sum_{S_i=\pm1}\prod_{b=1}^n e^{K_p (-1)^{e_b}R_b},
\end{equation}
where the interaction term for bond $b$, $R_b=\prod_j S_j^{G_{jb}}$,
is defined by the column $b$ of the matrix $G$, $e_b$ is the
corresponding bit in the vector $\mathbf{e}$, and the coupling
constant $K_p\equiv 1/T_p$ is the inverse  Nishimori
temperature\cite{Nishimori-1980,Nishimori-book}, $e^{-2K_p}=p/(1-p)$.

Similarly, for a given codeword $\mathbf{c}$, the probability of an
error degenerate with $\mathbf{e}+\mathbf{c}$ is proportional to
$Z_{\mathbf{e}+\mathbf{c}}(G;K_p)$.  Given the syndrome $\mathbf{s}=\mathbf{e}H^T$,
 the conditional probability that an error degenerate with $\mathbf{e}$  
actually happened can be written as the ratio\cite{Kovalev-Pryadko-SG-2015}%
\begin{equation}
  \label{eq:succ-decoding}
  P(\mathbf{e}|\mathbf{s})={Z_\mathbf{e}(G;K_p)\over
    \sum_\mathbf{c}Z_\mathbf{e+c}(G;K_p)}= {Z_\mathbf{e}(G;K_p)\over
    Z_\mathbf{e}(H^*;K_p)}, 
\end{equation}
where the sum in the denominator is proportional to the probability of
the syndrome $\mathbf{s}$ to happen.  In Eq.~(\ref{eq:succ-decoding}),
$H^*$ is a matrix dual of $H$, see Eq.~(\ref{eq:dual-matrix}); for
correct normalization, $H^*$ should be constructed from $G$ by adding
exactly $k$ rows corresponding to mutually non-degenerate codewords
$\mathbf{c}\in\mathcal{C}_H^\perp\setminus \mathcal{C}_G$.  The
conditional probability (\ref{eq:succ-decoding}), with
$\mathbf{e}=\mathbf{e}_\mathrm{max}(\mathbf{s})$ taken from the most
likely degeneracy class for the syndrome $\mathbf{s}$, is the
probability of successful ML decoding for the given syndrome.  One can
then calculate the average probability of successful ML decoding
\cite{Kovalev-Pryadko-SG-2015},
\begin{equation}
P_{\rm succ}(G,H;K,p)=[P(\mathbf{e}|\mathbf{e}H^T)]_p,\quad K=K_p,
\label{eq:psucc}
\end{equation}
where $[\,\boldsymbol\cdot\,]_p$ denotes the averaging over error
vectors (each set bit $e_b=1$ occurs independently with probability
$p$).

Notice that if we take a temperature away from the Nishimori line,
$T\neq T_p\equiv 1/K_p$, we are using a decoder with an incorrect $p$,
which would result in suboptimal
decoding\cite{Kovalev-Pryadko-SG-2015}.  For an infinite sequence of
codes $(G_t,H_t)$, $t\in\mathbb{N}$ with increasing distance, we
define the \emph{decodable} region on the $p$-$T$ plane as such where
\begin{equation}
\lim_{t\to\infty}P_{\rm succ}(G_t,H_t;K,p)=1.\label{eq:decodable-region}
\end{equation}
The overlap of the decodable region with the Nishimori line gives the
threshold error rate $p_c$ for ML decoding with the chosen sequence of
codes.  More generally, the extent of the decodable region away from
the Nishimori line can be seen as a measure of the decoding 
robustness.%

Generally, a code with the distance $d$ can correct any
$\lfloor (d-1)/2\rfloor$ errors.  If the errors on different (qu)bits
happen independently with probability $p$, a typical error has weight
asymptotically close to $p n$; the existence of a decodable region is
guaranteed only if the asymptotic relative distance $\delta=d/n$ is
finite.  Thus, in general, the decoding threshold satisfies
$p_c\ge \delta/2$.

Existence of a finite threshold for (quantum and classical) LDPC codes
with sublinear distance scaling where $\delta=0$ has been established
by two of us in Ref.~\onlinecite{Kovalev-Pryadko-FT-2013}.  The basic
reason for the existence of a threshold is that at small enough $p>0$,
likely error configurations can be decomposed into relatively small
connected clusters on the (qu)bit connectivity graph.  Specifically,
two (qu)bits are considered connected if there is a check (a
stabilizer generator) with the support including both positions.  For
an LDPC code (quantum or classical), the connectivity
graph has a bounded degree.  Then, formation of the connected error
clusters is described by the site percolation process on the
connectivity graph; it has a finite threshold\cite{Hammersley-1961}
$p_\mathrm{perc}\ge(\Delta-1)^{-1}$ for any graph with the maximum
degree $\Delta$.  Moreover, below this bound, the probability to
encounter a large cluster decreases exponentially with the cluster
size; this fact may be used to construct a syndrome-based
decoder\cite{Kovalev-Pryadko-FT-2013}.

More accurate lower bounds for decoding thresholds in different error
models (including phenomenological error model for syndrome
measurement errors) are given in
Ref.~\onlinecite{Dumer-Kovalev-Pryadko-bnd-2015}.  Consider CSS
codes whose generator matrices $G_x$ and $G_z$ have row weights not
exceeding some fixed $m$, and distance scaling logarithmically or faster with
$n$,%
\begin{equation}
d\ge D\ln n. \label{eq:log-d}
\end{equation}
Assuming independent $X$ and $Z$
errors with equal probabilities $p=p_X=p_Z$, the corresponding lower
bound reads
\begin{equation}
  \label{eq:decoding-bound-orig}
  2 [p(1-p)]^{1/2}\ge (m-1)^{-1}e^{-1/D}.
\end{equation}
With distance scaling like a power of $n$, $d\ge A n^{\alpha}$ with
$A,\alpha>0$, one should use $D=\infty$.  The bound
(\ref{eq:decoding-bound-orig}) was obtained by analyzing a minimum
energy decoder, which corresponds to $T=0$.

\subsection{Duality}
\label{sec:duality}

As demonstrated by Wegner\cite{Wegner-1971}, a general Ising model
with the partition function (\ref{eq:Z}) has a dual representation,
which is a generalization of
Kramers-Wannier\cite{Kramers-Wannier-1941} duality.  The same duality
has been first established in coding theory by
MacWilliams\cite{MacWilliams-1963} as a relation between weight
polynomials of two dual codes.  It is convenient to introduce a
generalized partition function,
\begin{equation}
    \label{eq:em-Z}
  Z_{\mathbf{e},\mathbf{m}}(G;K) 
  \equiv 
            \sum_{S_i=\pm1}\prod_{b=1}^n R_b^{m_b}e^{K(-1)^{e_b}R_b},
\end{equation}
that involves binary vectors of ``electric'' $\mathbf{e}$ and
``magnetic'' $\mathbf{m}$ charges.   Then, the duality reads 
\begin{equation}
  \label{eq:em-duality}
  Z_{\mathbf{e},\mathbf{m}}(G;K)
  =
  (-1)^{\mathbf{e}\cdot\mathbf{m}}Z_{\mathbf{m},\mathbf{e}}(G^*;K^*)\,A(K),        
\end{equation}
where an $r^*\times n$ matrix $G^*$ is the exact dual of $G$
(dimensions $r\times n$), see Eq.~(\ref{eq:dual-matrix}), $K^*$ is the
Kramers-Wannier dual of $K$, $\tanh K^*=e^{-2K}$, and the scaling
factor depends on the dimensions of the matrices,
\begin{equation}
  \label{eq:duality-scaling}
  A(K)=2^{r-r^*+\rank G^*}(\sinh K \cosh K)^{n/2}. 
\end{equation}
Notice that the electric charges in Eq.~(\ref{eq:em-Z}) define the
negative bonds as in Eq.~(\ref{eq:Z}), while the magnetic charges
select the bonds to be used in an average,%
\begin{equation}
  \label{eq:spin-average}
{Z_{\mathbf{e},\mathbf{m}}(G;K)\over Z_{\mathbf{e},\mathbf{0}}(G;K)}=\Bigl\langle \prod_{b=1}^n
  R_b^{m_b}\Bigr\rangle =\Bigl\langle \prod_{b=1}^n\prod_{j=1}^r
  S_j^{G_{jb}m_b}\Bigr\rangle,  
\end{equation}
which is the most general form of a spin correlation function that is
not identically zero\cite{Wegner-1971}.

\subsection{Quantum hypergraph-product codes}
In this work we specifically focus on the quantum hypergraph product
(QHP) codes\cite{Tillich-Zemor-2009,Tillich-Zemor-2014}, an infinite
family of quantum CSS codes which includes finite-rate LDPC codes with
distance scaling as a square root of the block length.  A general QHP
code is defined in terms of a pair of binary matrices $\mathcal{H}_1$
and $\mathcal{H}_2$ with dimensions $r_1\times n_1$ and
$r_2\times n_2$.  The corresponding stabilizer generators are formed
by two blocks constructed as Kronecker
products\cite{Kovalev-Pryadko-2012},
\begin{equation}
\begin{array}{c}
  \displaystyle
  G_{x}=(E_{2}\otimes\mathcal{H}_{1},\mathcal{H}_{2}\otimes E_{1}),\\
  \displaystyle
  G_{z}=(\mathcal{H}_{2}^{T}\otimes\widetilde{E}_{1},\widetilde{E}_{2}
  \otimes\mathcal{H}_{1}^{T}),
\end{array}\label{eq:Till}
\end{equation}
where $E_i$ and $\widetilde{E}_i$, $i=1,2$, are unit matrices of
dimensions given by $r_i$ and $n_i$; the matrices $G_x$ and $G_z$ have
$r_1r_2$ and $n_1n_2$ rows, respectively.  Clearly, the ansatz
(\ref{eq:Till}) guarantees that the rows of $G_x$ and $G_z$ are
orthogonal, $G_x G_z^T=0$.  The block length of such a quantum code is
the number of columns, $n\equiv r_2 n_1+r_1 n_2$.

We are using the construction originally proposed in
Ref.~\onlinecite{Tillich-Zemor-2009}, namely,
$\mathcal{H}_2=\mathcal{H}_1^T$, where $\mathcal{H}_1$ is assumed to
have a full row rank.  If the binary code with the check matrix
$\mathcal{H}_1$, $\mathcal{C}_{\mathcal{H}_1}^\perp$, has parameters
$[n_1,k_1,d_1]$, the corresponding QHP code has the
parameters\cite{Tillich-Zemor-2009,Tillich-Zemor-2014} $[[n,k,d]]$,
where $n=n_1^2+r_1^2$, $k=k_1^2$, $d=d_1$, and $r_1=n_1-k_1$.

We should mention that the QHP construction with
${\cal H}_2={\cal H}_1^T$ is weakly self-dual, meaning that the
matrices $G_x$ and $G_z$ in Eq.~(\ref{eq:Till}) can be transformed
into each other by row and column permutations.  As a result, in
particular, for any binary matrix ${\cal H}_1$, the decoding
probabilities (\ref{eq:psucc}) in $X$ and $Z$ sectors must coincide,
$P_\mathrm{succ}(G_x,G_z;K,p)=P_\mathrm{succ}(G_z,G_x;K,p)$.%

A family of quantum LDPC codes with distance scaling as a square root
of the block size can be obtained, e.g., by taking $\mathcal{H}_1$
from a random ensemble of classical LDPC codes, which are known to
have finite rates $k_1/n_1$ and finite relative distances $d_1/n_1$,
and removing any linearly-dependent rows.  We specifically consider
the ensemble $\mathbb{B}(\ell,m)$ of regular $(\ell,m)$-LDPC codes
with column weight $\ell$ and row weight $m$ originally introduced by
Gallager\cite{Gallager-1962,Gallager-book-1963}.  For each code in
$\mathbb{B}(\ell,m)$, its parity-check matrix $\mathcal{H}$ of size
$r\times n$ is divided into $\ell$ horizontal blocks
$\mathcal{H}_{1},\ldots,\mathcal{H}_{\ell}$ of size
$\frac{r}{\ell}\times n$.  Here the first block $\mathcal{H}_{1}$
consists of $m$ unit matrices of size
$\frac{r}{\ell}\times\frac{r}{\ell}$.  Any other block
$\mathcal{H}_{i}$ is obtained by some random permutation $\pi_{i}(n)$
of $n$ columns of $H_{1}$. Thus, all columns in each block $H_{i}$
have weight $1$. This ensemble achieves the best asymptotic distance
for a given designed code rate $1-\ell/m$ among the LDPC ensembles
studied to date\cite{Litsyn-Shevelev-2002}.  In practice, it often
happens that one or few rows of thus constructed ${\cal H}_1$ are
linearly dependent, which gives a code with a larger rate,
$R\equiv k/n\ge 1-\ell/m$.  It is easy to check, however, that
asymptotic at $n\to\infty$ rate equals the designed rate,
$R\to1-\ell/m$.  For the ensemble $\mathbb{B}(3,4)$ used in this work,
the asymptotic relative distance $\delta=d/n$ is
$\delta_{3,4}\approx 0.112$ \cite{Litsyn-Shevelev-2002}.

For brevity, we will refer to QHP codes constructed using the matrices
${\cal H}_2={\cal H}_1^T$ from Gallager $\mathbb{B}(\ell,m)$ ensemble
(with any linearly dependent rows dropped) as $(\ell,m)$ QHP codes.

\section{Analytical bounds}
\label{sec:qec}

Partition function (\ref{eq:Z}) scales exponentially with the system
size; one more commonly works with the corresponding logarithm, the
(dimensionless) free energy
\begin{equation}
  \label{eq:free-energy}
  F_\mathbf{e}(G;K)=-\ln Z_\mathbf{e}(G;K), 
\end{equation}
which is an extensive quantity, meaning that it scales linearly with
the system size.  Alternatively, one can also use the free energy
density (per bond), $f_\mathbf{e}(G;K)=F_\mathbf{e}(G;K)/n$, which
usually has a well-defined thermodynamical limit.  The logarithm of
the ML decoding probability (\ref{eq:succ-decoding}), up to a sign,
equals the \emph{homological difference},%
\begin{equation}
  \label{eq:homological-difference}
  \Delta F_\mathbf{e}(G,H;K)\equiv
  F_\mathbf{e}(G;K)-F_\mathbf{e}(H^*;K).
\end{equation}
At $\mathbf{e}=0$, this quantity 
satisfies the inequalities
\begin{equation}
  0\le \Delta F_\mathbf{0}(G,H;K)\le k\ln2 ,\label{eq:homo-generic-bounds}
\end{equation}
where the lower and the upper bounds are saturated, respectively, in
the limits of zero and infinite temperatures.  Combining duality
(\ref{eq:em-duality}) with
Griffiths-Kelly-Sherman\cite{Griffiths-1967,Kelly-Sherman-1968} (GKS)
inequalities for spin averages we also obtain
\begin{equation}
  \label{eq:homo-diff-disorder}
\Delta F_\mathbf{e}(G,H;K)- \Delta F_\mathbf{0}(G,H;K)\ge 0.
\end{equation}
In addition, also at $\mathbf{e}=0$, the duality (\ref{eq:em-duality})
gives
  \begin{equation}
    \Delta
    F_\mathbf{0}(G,H;K)=k\ln2- \Delta
    F_\mathbf{0}(H,G;K^*),\label{eq:homo-duality}  
\end{equation}
where $\tanh K^*=e^{-2K}$, and $k$ is the dimension of the CSS code,
see Eq.~(\ref{eq:k-defined}).  The proof of these expressions is given
in Appendix \ref{app:homo-diff}.

The relation of the homological difference averaged over the disorder,
$[\Delta F_\mathbf{e}]_p $, and the corresponding quantity normalized
per unit bond,
$[\Delta f_\mathbf{e}]_{p}\equiv [\Delta F_\mathbf{e}(G,H;K)]_p/n$, to
decoding with asymptotic probability one, see
Eq.~(\ref{eq:decodable-region}), is given by the following Lemma
(proved in App.\ \ref{app:lemma1}). 
\begin{restatable}{lemma}{thcodesF}
  \label{th:codes-F}
  For a sequence of quantum CSS codes defined by pairs of matrices
  $(G_t,H_t)$, $t\in\mathbb{N}$, where $G_t H_t^T=0$, given a finite
  $K>0$ and an error probability $p\ge0$,\\(a)
  $\lim_{t\to\infty}[\Delta F_{\bf e}(G_t,H_t;K)]_p= 0$ implies the point $(p,K)$
  to be in the decodable region; \\(b)
  $\liminf_{t\to\infty}[\Delta f_\mathbf{e}(G_t,H_t;K)]_p>0$ implies the point
  $(p,K)$ to be outside of the decodable region.
\end{restatable}

\subsection{Lower bound for decodable region}
Here, we use part (a) of Lemma \ref{th:codes-F} to establish an
existence bound for the decodable region.  Specifically, we construct
an upper bound for $[\Delta F_\mathbf{e}(G,H;K)]_p\ge0$ in a finite
system, and use it to show the existence of a non-trivial region where
$[\Delta F_\mathbf{e}]_p\to0$, as long as the distance scales
logarithmically or faster with the block length $n$, see Eq.~(\ref{eq:log-d}).
In  Appendix \ref{sec:lts-convergence-codes} we prove:
\begin{restatable}{theorem}{ltsconvergencecodes}
  \label{th:lts-convergence-codes}
  Consider a sequence of quantum CSS codes $\mathcal{Q}(G_t,H_t)$,
  $t\in\mathbb{N}$, of increasing lengths $n_t$, where row weights of
  each $G_t$ and $H_t$ do not exceed a fixed $m$, and the code
  distances $d_{t}\ge D\ln n_t$, with some $D>0$.  Then the sequence
  $\Delta F_t\equiv [\Delta
  F_\mathbf{e}(G_t,H_t;K)]_p$, $t\in\mathbb{N}$,
  converges to zero in the region
  \begin{equation}
    (m-1)[e^{-2K}(1-p)+e^{2K}p]<e^{-1/D}.\label{eq:low-T-low-p-region-two}
\end{equation}
\end{restatable}
The rightmost point of this region, the maximum value
$p=p_\mathrm{bnd}$ where Eq.~(\ref{eq:low-T-low-p-region-two}) has a
solution, satisfies the equation
$2(m-1)[p_\mathrm{bnd}(1-p_\mathrm{bnd})]^{1/2}=e^{-1/D}$.  The same
bound was obtained previously in
Ref.~\onlinecite{Dumer-Kovalev-Pryadko-bnd-2015} using estimates based
on minimum-energy decoding which corresponds to $T=0$.  Thus, present
bound does not improve the existing lower bound for ML decoding
threshold.

Further, the entire region (\ref{eq:low-T-low-p-region-two}) lies at
temperatures $T=1/K$ above the Nishimori line (see Fig.~\ref{fig1}).
In particular, at the right-most cusp of this region, the temperature
$T_\mathrm{bnd}=1/K_\mathrm{bnd}$ is exactly twice the Nishimori
temperature at $p_\mathrm{bnd}$.  The importance of Theorem
\ref{th:lts-convergence-codes} is that we got a sense of the
robustness of suboptimal decoding, where the ML decoder assumes a
value of $p$ larger than the actual one.

\subsection{Upper temperature bound for decodable region}
\label{sec:bnd-upper}

Here we combine part (b) of Lemma (\ref{th:codes-F}) with duality
(\ref{eq:em-duality}) to establish an upper temperature bound for the
decodable region for a sequence of codes with asymptotic rate $R$.  We
first argue that existence of a low-temperature \emph{homological}
region where $\Delta f_\mathbf{0}(G,H;K)\to0$, by duality, implies the
existence of a high-temperature \emph{dual homological} region where
$\Delta f_\mathbf{0}(H,G;K)\to R\ln2$, and thus
$[\Delta f_\mathbf{e}(H,G;K)]_p\ge R\ln2$ at any $p\ge0$.  Further,
the derivative of $f_\mathbf{e}$ with respect to $K$ is just the
energy per bond, with negative sign; its magnitude does not exceed
one.  Therefore, there should be some minimal distance between the
upper temperature bound $K_0^{(1)}(G,H)$ of the homological region and
the lower temperature bound $K_0^{(2)}(H,G)$ at $p=0$ of the dual
homological region.  This gives an upper temperature bound for the
homological region.  By part (b) of Lemma (\ref{th:codes-F}), the same
bound also works as an upper bound for the decodable region at any
$p>0$.  These arguments give (see App.\ \ref{app:upper-decod}):%
\begin{restatable}{theorem}{upperhomobound}
  \label{th:upper-homo-bound}
  Consider a sequence of CSS codes defined by pairs of finite binary
  matrices with mutually orthogonal rows, $G_tH_t^T=0$,
  $t\in\mathbb{N}$, where row weights of $G_t$ and $H_t$ do not exceed
  a fixed $m$, the sequence of CSS distances
  $d_t=\max(d_{H_t},d_{G_t})$ is strictly increasing with $t$,
  $d_{t+1}>d_t$, and the sequence of rates $R_t\equiv k_t/n_t$
  converges, $\lim_{t\to\infty} R_t=R$.  Then, assuming equal
  probabilities of $X$ and $Z$ errors, the upper temperature boundary
  of the decodable region, $T_{\rm max}=1/K_{\rm max}$, satisfies the
  inequality
  \begin{equation}
    \label{eq:ineq-decod}
    K_{\rm max}-K_{\rm max}^*\ge R\ln2.
  \end{equation}
\end{restatable}

Explicitly, this gives an upper temperature bound for the location of
the ML-decodable region for any CSS code family with asymptotic rate
$R$,
\begin{equation}
  \label{eq:inec-decod-solved}
e^{2K_{\rm max}}\ge {1+r  +\sqrt{(1+r)^2+r}\over 2},\quad r\equiv 2^{2R}\ge1.
\end{equation}
In the case $R=0$ this bound corresponds to the self-dual point, which
equals to the upper bound of the decodable region of the
square-lattice toric code (ferromagnetic phase of the square-lattice
Ising model).

\begin{figure*}[thbp]
\includegraphics[width=0.75\textwidth]{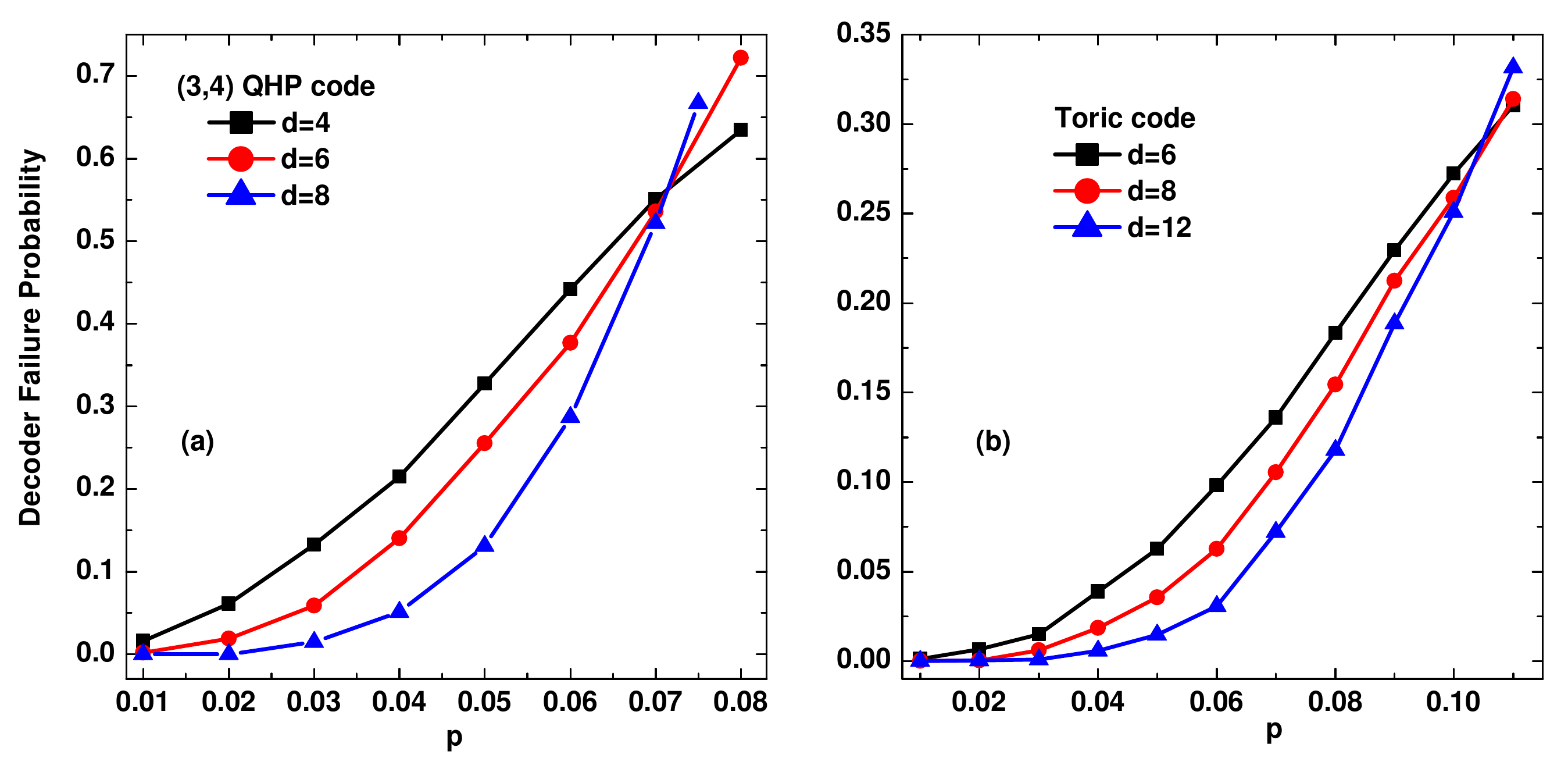}
\caption{(Color online) Decoder failure probability as a function of
  bit-flip error probability for three $(3,4)$ QHP codes (left) and
  the rotated toric codes (right). The decoding was performed over
  $1024$ error realizations for $(3,4)$ QHP codes and over $4096$
  error realizations for toric codes.  The corresponding decoding
  pseudothreshold is close to $p=7.0\%$ for $(3,4)$ QHP codes and to
  $10.4\%$ for toric codes.}
\label{fig2}
\end{figure*}



\section{Numerical results}
\label{sec:numerics}

\subsection{Justification}
We first note that numerically, it is only possible to analyze systems
of finite size.  Numerical techniques used for predicting asymptotic
large-size properties, such as finite-size scaling, are only good as
long as such properties exist and change with the system size in a
regular manner.  For example, even though we know the existence of a
non-zero decoding threshold, it is not a priori clear that the
finite-size data would show a well defined crossing point, as seen on
Fig.~\ref{fig2}.

Similarly, a well-defined thermodynamical limit is known to exist for
bulk quantities like magnetization or specific heat for Ising models
on lattices that are local in $D$ dimensions, simply because the
corrections due to the boundary scale as the surface area, which
scales as a sublinear power of the
volume\cite{Griffiths-results-1972}.  Well-defined infinite-size limit
(although not necessarily universal) also exists if one considers a
sequence of models on increasing subgraphs of an infinite graph, where
the boundary spins are either ``free'' (the couplings connecting them
to outside are set to zero), or ``wired'' (the outside couplings are
set to infinity).  The existence of a thermodynamical limit in each of
these cases follows from the GKS
inequalities\cite{Griffiths-1967,Kelly-Sherman-1968}, which require
that spin correlations change monotonously with the system size
(increase for wired and decrease for free boundary conditions). 

The problem we are considering is different from either case, as a
sequence of matrices $G_t$ (or the corresponding bipartite graphs)
defines a sequence of finite few-body Ising models without boundaries.
Further, the finite asymptotic rate of the considered code family
guarantees the
absence\cite{Bravyi-Terhal-2009,Bravyi-Poulin-Terhal-2010} of a
$D$-dimensional layout of the qubits with local stabilizer generators,
at any finite $D$.  The only rigorous result, proved in a companion
paper\cite{Jiang-Kovalev-Dumer-Pryadko-2018}, is that a well defined
limit for average free energy density for $(\ell,m)$ QHP codes exists
for any $p$ in a finite region around the infinite temperature and, by
duality, for $p=0$, in a finite region around the zero temperature.
This follows from the absolute convergence of the corresponding
high-temperature series (HTS) established using the bound on
high-order cumulants\cite{Feray-Meliot-Nikeghbali-2013}, and from the
fact that a large random bipartite graph with vertex degrees $\ell$
and $m$ has few short cycles.  The local Benjamini-Schramm
limit\cite{Benjamini-Schramm-2001} of such a graph is a bipartite
tree, meaning that the asymptotic coefficients of the high-temperature
series expansion to any finite order can be computed by analyzing only
the clusters present when ${\cal H}_1$ in Eq.~(\ref{eq:Till})
corresponds to such a tree.  The corresponding argument is a direct
generalization of that in
Refs.~\onlinecite{Borgs-Chayes-Kahn-Lovasz-2013,Lovasz-2016}, where
the existence of a well defined limit for free energy density was
analyzed for general models with up to two-body interactions.

One consequence of this argument is that in the asymptotic limit, we
do not expect much difference between the use of matrices ${\cal H}_1$
from the full Gallager $\mathbb{B}(\ell,m)$ ensemble, and the
corresponding subset where for each size we pick only the matrices
which result in the largest distance $d_1$ of the classical code
${\cal C}_{{\cal H}_1}^\perp$.  On the other hand, we expect that the
use of such matrices should significantly improve the convergence in
the high- and low-temperature regions where the corresponding series
converge: with larger distance, a larger number of coefficients of the
series would match those for the infinite-size system.

Unfortunately, even though the corresponding series can be
analytically continued beyond the convergence radius, this does not
guarantee the existence of a well-defined limit for thermodynamical
quantities at all temperatures, as would be required to formally
justify the use of finite size scaling.  Therefore, numerical results
presented in the following sections represent numerical trends in
systems of relatively small size; they do not necessarily guarantee
the existence of well defined transition(s).

\subsection{Approximate minimum-weight decoding}
\label{sec:mw-decoding}

To obtain an empirical lower bound for the ML decoding threshold, we
constructed a cluster-based decoder using the approach suggested in
Refs.~\cite{Kovalev-Pryadko-FT-2013,Dumer-Kovalev-Pryadko-bnd-2015}
(see Sec.~\ref{sec:ml-decoding}).  Specifically, given the syndrome
vector $\mathbf{s}$, we construct a list of \emph{irreducible
  clusters} up to the chosen cut-off weight $w_{1}$.  Each irreducible
cluster should correct some syndrome bits without introducing new
ones, and it should not contain a subcluster with the same property.
As explained in Sec.~\ref{sec:ml-decoding}, with an LDPC code where
the stabilizer weight is bounded, and for large enough $w_\mathrm{1}$,
we expect this list with high probability to include all clusters
present in the connected-cluster decomposition of the actual error.
The actual decoding is done by solving a minimum-weight set cover
problem: among the subsets of the cluster list with the property that
every non-zero syndrome bit be covered exactly once, we want to find
such that the sum of the cluster weights be minimal.  This latter
problem is solved in two steps: first, by running the
\texttt{LinearProgramming} over integers in
Mathematica\cite{mathematica} to arrive at a valid solution with a
reasonably small weight, and then by trying to minimize the weight
further with the help of a precomputed list of non-trivial irreducible
codewords\cite{Dumer-Kovalev-Pryadko-bnd-2015}.  In our calculations,
for each disorder realization we generated irreducible clusters of
weight up to $w_1=10$, and, for each code, the list of irreducible
codewords of weight up to $w_2=19$.  

Without the limits on the clusters' and codewords' weights, this
procedure would be equivalent to minimum-weight decoding.
Unfortunately, the corresponding complexity grows prohibitively
(exponentially with the size of the code).  Nevertheless, for smaller
codes we were able to choose large enough $w_1$ and $w_2$ to estimate
the minimum-weight decoding threshold, as seen from the convergence of
the corresponding decoding probabilities.  

The decoding complexity is determined by the sum of those for the
construction of the cluster list and for solving the weighted set
cover problem.  The construction of the cluster list was analyzed in
detail in Refs.~\onlinecite{Dumer-Kovalev-Pryadko-2014,%
  Dumer-Kovalev-Pryadko-ISIT-2016,Dumer-Kovalev-Pryadko-IEEE-2017}.
In particular, if the maximum weight of a stabilizer generator is $m$,
the corresponding complexity is $N_1\sim n (m-1)^{w_1-1}$.  At small
enough $p$, the probability of a large cluster decays exponentially
with its weight.  Thus, in most cases, maximum cluster size scales
logarithmically with the code length $n$, and a sufficient cluster
list can be prepared with the cost polynomial in $n$.

On the other hand, the weighted set cover problem is
NP-complete\cite{Garey-Johnson-book-1979}; the corresponding cost is
exponential in the length $L$ of the cluster list.  Generally, this
problem is equivalent to an integer linear programming (LP) problem.
To find a valid (but not necessarily the minimal) solution, we use a
call to the built-in Mathematica function \texttt{LinearProgramming}.
While the details of its implementation are proprietary, it is our
understanding that an integer solution is found by first solving the
corresponding problem over reals using an algorithm with polynomial
complexity, and then finding the nearest integer point in the LP
polytope.  With rare exception (few instances over the entire set of
our simulations where we had to record decoder failure due to
calculation time-out), \texttt{LinearProgramming} returns a valid
solution $\mathbf{e}$ which satisfies the constraints but does not
necessarily have the smallest weight.

To reduce the weight further, we used a version of the approach used
previously\cite{Dumer-Kovalev-Pryadko-bnd-2015} to construct an
analytical bound for minimum-energy decoding threshold.  Notice that
the minimum-weight (same as minimum-energy) solution
$\mathbf{e}_\mathbf{min}$ produces the same syndrome as $\mathbf{e}$,
thus $\Delta \mathbf{e}=\mathbf{e}-\mathbf{e}_\mathrm{min}$ produces
the zero syndrome, and in general can be decomposed into a sum of
\emph{irreducible codewords}\cite{Dumer-Kovalev-Pryadko-bnd-2015},
$\Delta\mathbf{e}=\mathbf{e}_1+\ldots +\mathbf{e}_s$, such (a) that
the supports of different $\mathbf{e}_j$ do not overlap, (b) each of
$\mathbf{e}_j$ is a valid codeword, in the sense that it produces a
zero syndrome, and (c) any $\mathbf{e}_j$ cannot be decomposed further
into a sum of non-overlapping codewords.  Such a decomposition is not
necessarily unique.  It is easy to
see\cite{Dumer-Kovalev-Pryadko-bnd-2015} that weight of
$\mathbf{e}+\mathbf{e}_j$ for any $j\in\{1,2,\ldots,s\}$ must not
exceed that of $\mathbf{e}$.  Thus, if we have a list of all
non-trivial, $\mathbf{c}_j\not\simeq \mathbf{0}$, irreducible
codewords, the equivalence class of the minimum-weight solution can be
found by adding those $\mathbf{c}_j$ that reduce the weight of
$\mathbf{e}$, until the weight can no longer be reduced.

Notice that with a complete list of non-trivial irreducible codewords,
the degeneracy class of the minimum weight solution can be correctly
identified from any vector $\mathbf{e}$ which produces the correct
syndrome.  In practice, since the weights of the irreducible codewords
in our list are limited, the decoding success probability increases
with the reduced weight of the initial vector $\mathbf{e}$.  

Overall, for each code in our simulations, the majority of the
computational time was spent on preparing the list of non-trivial
irreducible codewords with weights $w\le w_2=19$.


The results of the described threshold simulations are presented in
Fig.~\ref{fig2} (a) and (b), respectively, for $(3,4)$ QHP codes and
for the toric codes.  More precisely, in Fig.~\ref{fig2}(a), we show
the fraction of decoder failures for QHP codes constructed from three
large-distance classical codes from Gallager ${\mathbb{B}}(3,4)$
ensemble, for different values of bit flip probability $p$.  The codes
used have parameters $[[80,16,4]]$, $[[356,36,6]]$, and
$[[832,64,8]]$; they were constructed from binary codes with
parameters $[8,4,4]$, $[16,6,6]$, and $[24,8,8]$.  Code
$[[1921,121,10]]$ obtained from the binary code $[36,11,10]$ turned
out too large for the present decoding technique; the corresponding
data is not included in Fig.~\ref{fig2}(a).

In our calculations, we used $w_1=10$ and $w_2=19$, which was
sufficient for convergence of the average decoding probability for
$p\le 0.08$ used in the simulations.  The well defined crossing point
in Fig.~\ref{fig1}(a) indicates a (pseudo)threshold for decoding of
$(3,4)$ QHP codes in the vicinity of $7.0\%$.  Convergence of the
average decoding probability with increasing $w_1$ and $w_2$ is an
indication that this value is a good estimate of the minimum-weight
decoding threshold.

For comparison, in Fig.~\ref{fig2}(b), we show the corresponding
results for the rotated toric codes\cite{Bombin-2007} with the
parameters $[[d^2,2,d]]$, with $d=6$, $8$ and $12$, where the crossing
point is close to $ 10.4\%$, the minimum-weight decoder threshold
obtained using the minimum-weight matching
algorithm\cite{Dennis-Kitaev-Landahl-Preskill-2002}.

Notice that both for $(3,4)$ QHPs and for the toric code, the obtained
threshold estimates are much larger than the corresponding
analytical lower bounds from
Ref.~\onlinecite{Dumer-Kovalev-Pryadko-bnd-2015}, $0.70\%$ and
$2.8\%$, respectively.

\subsection{Monte Carlo simulations and the phase diagram}

In this section we analyze numerically the low-disorder portion of the
phase diagram of the two random-bond Ising models corresponding to the
ML decoding of $(3,4)$ QHP codes with i.i.d.\ bit-flip errors.  For a
CSS code with generators $G=G_x$ and $H=G_z$, the corresponding Ising
models have the free energies $F_\mathbf{e}(G;K)$ and
$F_\mathbf{e}(H^*;K)$, see Eqs.~(\ref{eq:Z}) and
(\ref{eq:free-energy}), where $\mathbf{e}$ is the binary error vector
whose non-zero bits indicate the flipped bonds, and $K=1/T$ is the
inverse temperature.  These models, respectively, correspond to the
numerator and the denominator of the conditional ML decoding
probability (\ref{eq:succ-decoding}).

The parameters of the four $(3,4)$ QHP codes used in the simulations
are described in the previous section.  For simulation efficiency, we
attempted to minimize the weights of the rows of the matrices $H^*$.
To this end, starting with the matrix $G'=G$, we added one row at a
time, corresponding to one of the minimum-weight vectors in
$\mathcal{C}_H^\perp\setminus \mathcal{C}_{G'}$, where $G'$ is the
previously constructed matrix.  As a result, the row weights of each
matrix $H^*$ did not exceed $\max(7,d_G)$.

To calculate the averages, we
performed feedback optimized parallel tempering Monte Carlo
simulations\cite{Katzgraber-Trebst-Huse-Troyer-2006,Kubica-etal-color-2017},
as well as  the usual simulated annealing.  In both cases we used standard Metropolis updates.

For both models, the observed scaling of the height of the specific
heat maxima with $n$, and the hysteresis which we could not eliminate
for larger codes, are consistent with the discontinuous transitions.
We also observe that the use of the parallel tempering method does not
improve the convergence significantly; we attribute this to the
discontinuity of the phase transition.

Samples of the computed specific heat (per bond) for the (3,4) QHP
models, $C(T)=(\langle E^2\rangle-\langle E\rangle^2)/(n T^2)$, where
$E$ is the energy, $n$ is the number of bonds, and $T$ is the
temperature, are shown in Fig.~\ref{fig_cv} (comparing different
values of $p$ separately for distances $d=4$, $6$, and $8$).  The
specific heat values shown in Fig.~\ref{fig4} have been additionally
divided by the number of bonds $n$; the corresponding maximum values
are weakly increasing with the code distance for $p=0$, see
Fig.~\ref{fig4}(a), and weakly decreasing for $p=2\%$ and $10\%$, see
Figs.~\ref{fig4}(b) and \ref{fig4}(c).  Such a slow dependence on the
system size is
consistent with a 1st order transition, where one expects
$C(T_c)\propto n$.  In Fig.~\ref{fig4}(c) we also compare the data
obtained using parallel tempering and the usual annealing.  These data
were obtained after $1\times 10^7$ Monte Carlo sweeps for QHP codes of
distance $d=4$ and $d=6$, and $5\times 10^7$ Monte Carlo sweeps for
codes of distances $d=8$ and $d=10$, for each of 128 (in some cases 256) realizations of
disorder at every $p$.

\begin{figure*}[htbp]
\includegraphics[width=\textwidth]{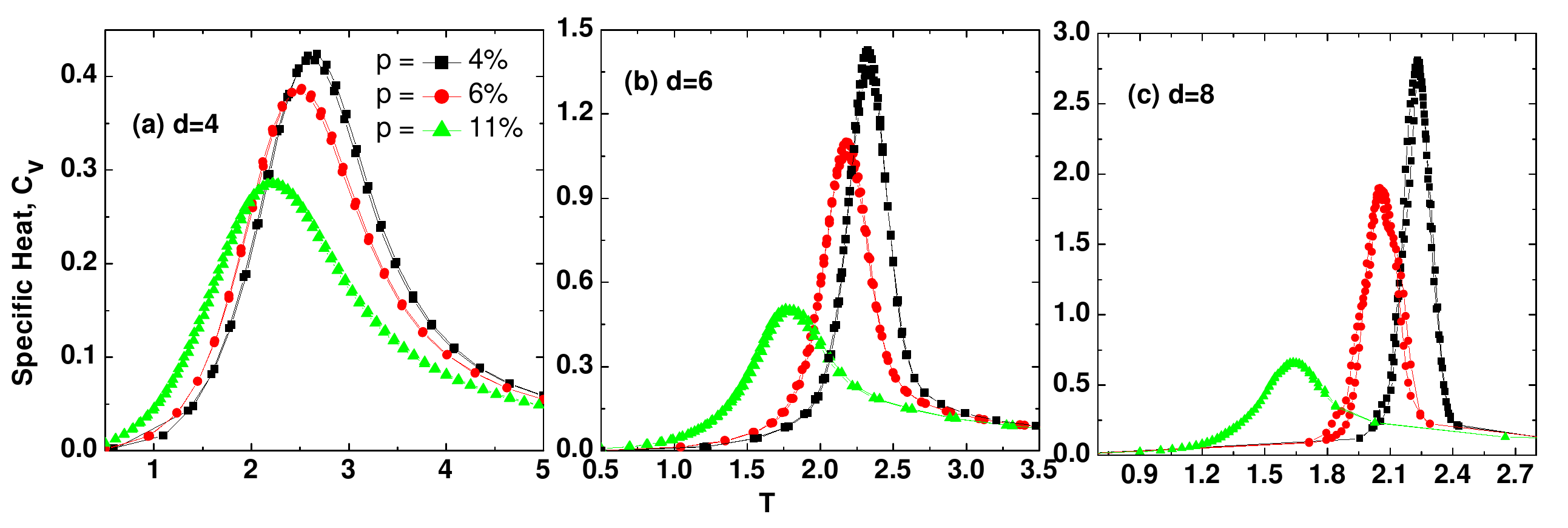}
\caption{(Color online) Specific heat $C$ vs.\ dimensionless
  temperature $T$ for $(3,4)$ QHP models with distances (a) $d=4$, (b)
  $d=6$, and (c)~$d=8$, at $p$ values as indicated on
  panel (a).  Each curve contains data points from
  the feedback optimized parallel tempering simulation where ordered and disordered configurations are used as initial states.  The peak positions are extrapolated to infinite
  distance to obtain the transition temperatures, see
  Fig.~\ref{fig5}.}
\label{fig_cv}
\end{figure*}
\begin{figure*}[htbp]
\includegraphics[width=\textwidth]{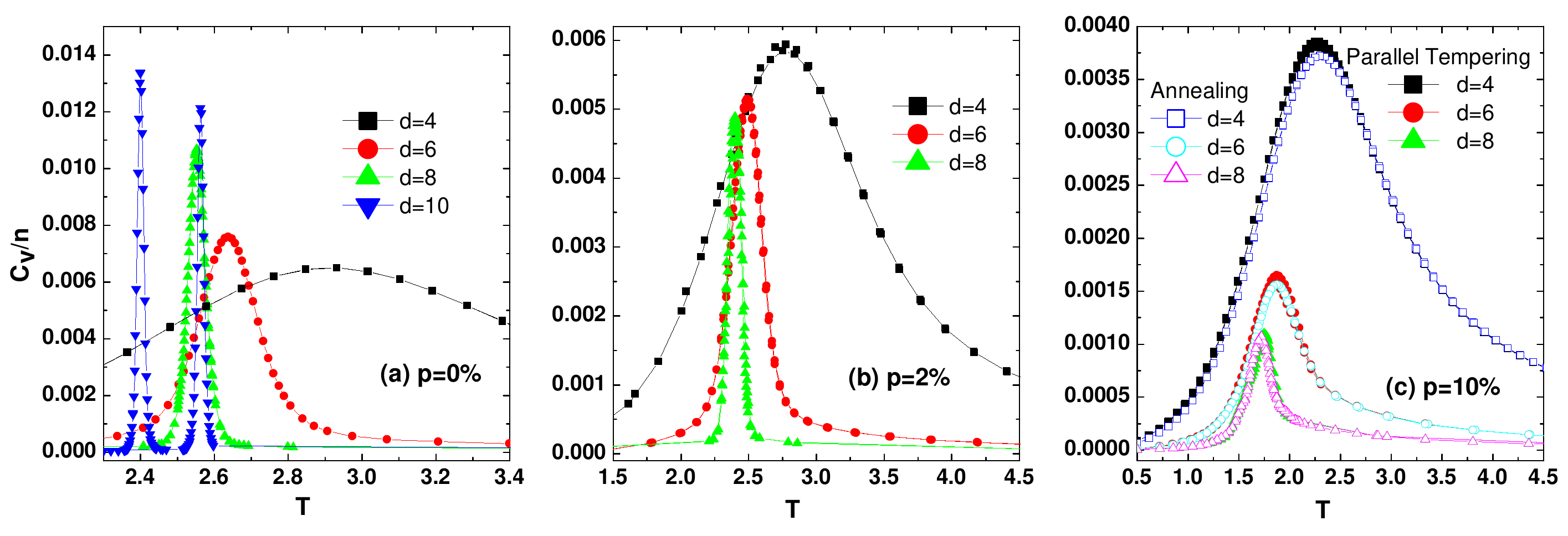}
\caption{(Color online) Specific heat $C$ divided by the number of
  bonds $n$ vs.\ dimensionless temperature $T$ for $(3,4)$ QHP models
  at (a)~$p=0$, (b)~$p=2\%$, and (c) $p=10\%$, with code distances as
  indicated in the captions.   Each curve contains data points from
  the feedback optimized parallel tempering simulation. 
  The annealing plot contains data points from upward and downward
  temperature sweeps. Relatively weak variation
  of the peak height with system size is indicative of a discontinuous
  transition.  Open symbols in plot (c) show the data obtained with
  annealing, which agrees with the parallel tempering data (filled
  symbols).}
\label{fig4}
\end{figure*}

\begin{figure}[htbp]
\includegraphics[width=\columnwidth]{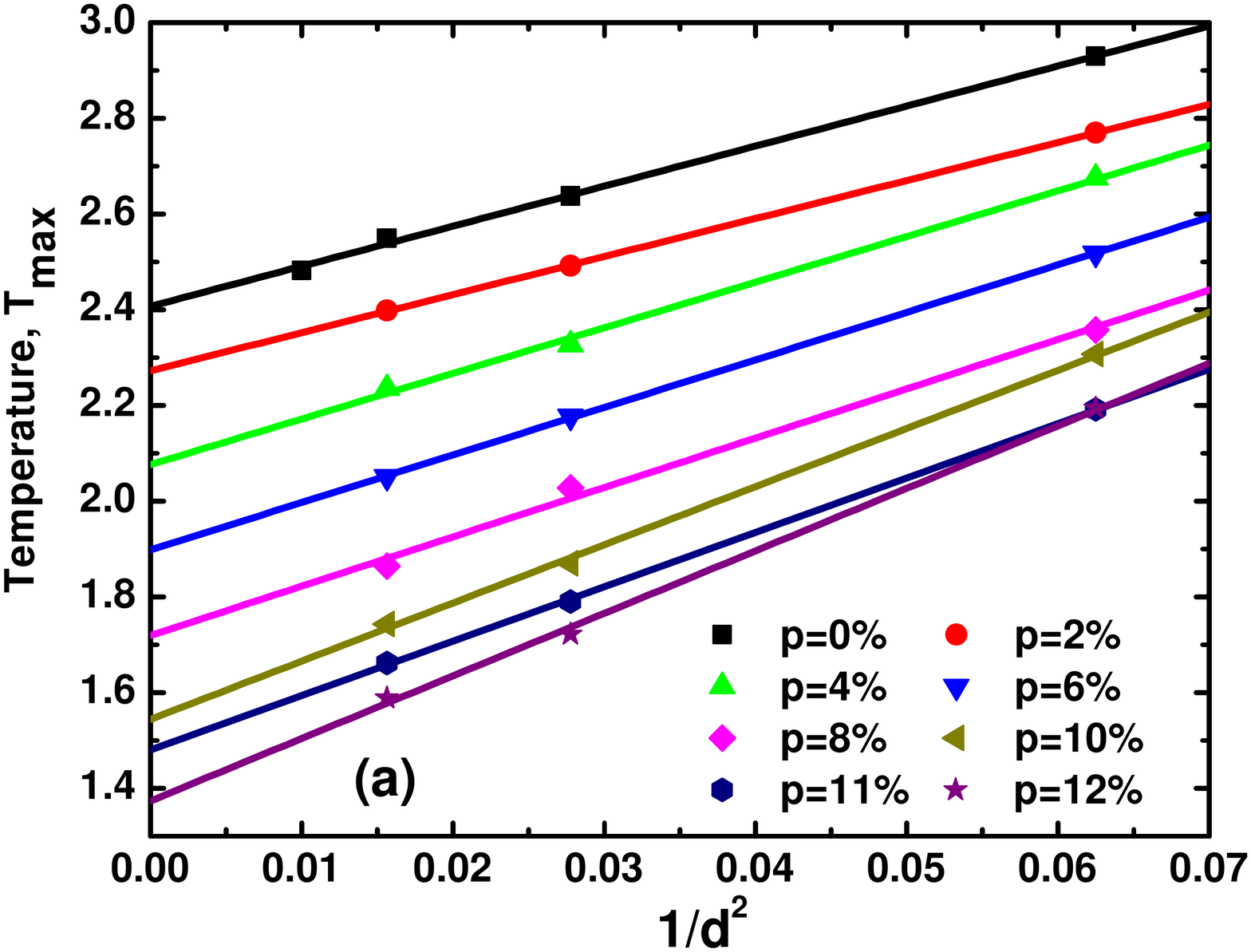}\\
\includegraphics[width=\columnwidth]{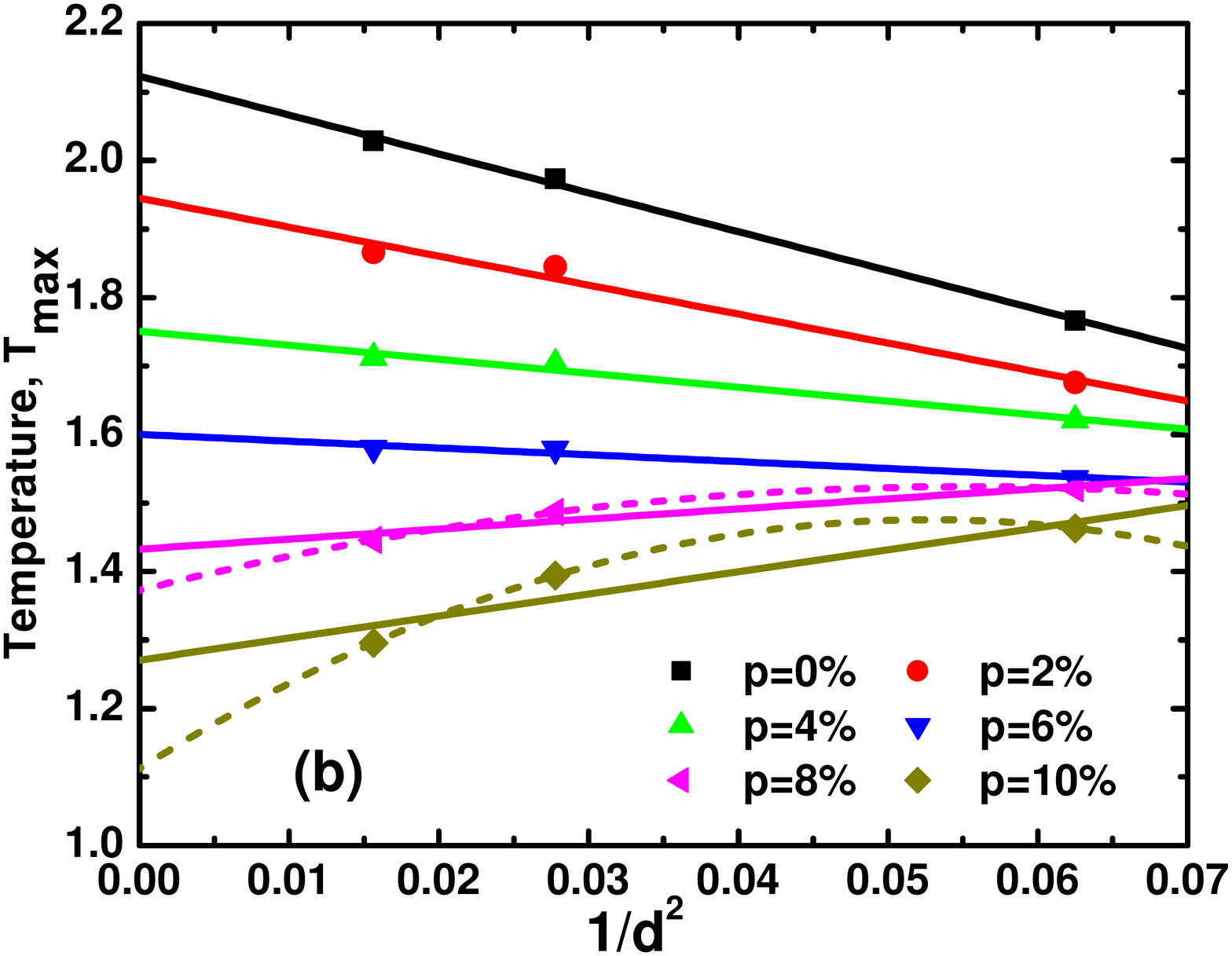}
\caption{(Color online) (a) Finite size scaling of the dimensionless
  temperatures $T_\mathrm{max}$ where specific heat reaches the
  maximum for $(3,4)$ QHP models vs.\ the inverse square of the code
  distance, with the fraction of flipped bonds $p$ as indicated.  For
  the point $p=0$ and $d=10$ where we could not eliminate the
  hysteresis, the average temperature was used.  (b) Same for the dual
  $(3,4)$ QHP models.  To accommodate the increased curvature, for
  $p=0.08$ and $0.10$ we also used parabolic fits,
  $T_\mathrm{max}=T_c+A/d^2+B/d^4$, which results in a significantly
  reduced extrapolated value of $T_c$ for $p=0.10$.  The $T_c$ values
  from parabolic fits are shown in Fig.~\ref{fig1} with open red
  symbols.}
\label{fig5}
\end{figure}

When the positions of the specific heat maxima are plotted as a
function of $1/d^2$ (asymptotically, $d^2\propto n$, although such a
relation does not hold for the small codes used in the simulations),
the corresponding points are seated close to a straight line, see
Fig.~\ref{fig5}(a).  Respectively, we used the linear fit
$T_\mathrm{max}(d,p)=T_c(p)+A/d^2$ to extrapolate our finite-size data
and extract more accurate critical point of the transition, $T_c(p)$,
as a function of the flipped bond probability $p$.  The resulting
phase boundary is shown in Fig.~\ref{fig1} with solid blue circles,
along with the solid blue line which is the linear fit to the data.
The fit indicates that for $p\le 0.12$, the phase transition
temperature $T_c(p)$ is approximately linear in $p$.  It is also clear
from Fig.~\ref{fig1} that at every $p$, the phase boundary for this
model is higher than the corresponding line for the square-lattice
Ising model, plotted with dot-dashed line using the data from
Ref.~\onlinecite{Thomas-Katzgraber-2011}.

The analysis for the dual $(3,4)$ QHP models was performed similarly
(specific heat data not shown).  The positions of the specific heat
maxima as a function of $1/d^2$ for different values of $p$ are shown
in Fig.~5(b), along with the corresponding linear fits.  Notice that
the points at $p=10\%$ show significant curvature which cannot be
attributed to the statistical errors alone.  By this reason we also
tried a parabolic fit, which resulted in a substantially lower
extrapolated $T_c=1.11$ compared with $1.27\pm0.05$ from the linear
fit.  In comparison, at $p=8\%$, parabolic fit gives $T_c=1.37$, which
is not as significantly reduced compared to the linear fit result of
$1.43\pm0.02$.

The extrapolated positions of the specific heat maxima are plotted in
Fig.~\ref{fig1} with solid red boxes, along with a solid red
line which is the ad hoc linear fit to the data.  (The
two extrapolated values obtained from parabolic fits in
Fig.~\ref{fig5}(b) are shown in Fig.~\ref{fig1} with open red
boxes.)  The corresponding line is approximately parallel to that
for the (3,4) QHP model.  As expected (see Sec.~\ref{sec:duality}),
the points at $p=0$ are located close to mutually dual positions.  For
the dual model, the extrapolation gives 
$T_c(H^*,p=0)\approx 2.12$, which is close to $T_c^*(G,p=0)=2.14$
obtained from $T_c(G,p=0)\approx 2.41$.

Empirically, the transition temperatures in the two dual models are
different.  Under this
condition\cite{Jiang-Kovalev-Dumer-Pryadko-2018} $\biglb($more
precisely, assuming that large-system free energy density
$[f_\mathbf{e}(H^*,K)]_p$ be non-singular at and below the
lowest-temperature singular point of $[f_\mathbf{e}(G,K)]_p\bigrb)$
the transition temperature of the dual (3,4) QHP model should coincide
with the homological transition, where
$[\Delta f_\mathbf{e}(G,H;K)]_p$ reaches the lower bound of $0$ [cf.\
Eqs.~(\ref{eq:homo-generic-bounds}) and
(\ref{eq:homo-diff-disorder})].  Above this temperature,
$[\Delta f_\mathbf{e}]_p>0$.  Thus, according to part (b) of Lemma
\ref{th:codes-F}, the critical point $T_c(H^*,p)$ of the dual model
(red squares in Fig.~\ref{fig1}) gives an upper bound for the
decodable phase of $(3,4)$ QHP codes.

Fig.~\ref{fig1} also shows several analytical bounds for the decodable
region.  Magenta-shaded region corresponds to the lower bound for the
decodable region given by Eq.~(\ref{eq:low-T-low-p-region-two}) with
$m=7$.  Its rightmost point is at the same $p$ as the energy-based
analytical bound from Ref.~\cite{Dumer-Kovalev-Pryadko-bnd-2015}, the
corresponding point is indicated on the horizontal axis by the magenta
arrow.  The lower bound for the decodable region (pseudothreshold for
energy-based decoding) is shown with the blue vertical arrow.
Finally, a pair of gray arrows separated by the gray bar on the
vertical axis show the bound $T_\mathrm{max}$ from Theorem
\ref{th:upper-homo-bound} and the corresponding dual temperature,
$T_\mathrm{max}^*>T_\mathrm{max}$.  As
expected\cite{Jiang-Kovalev-Dumer-Pryadko-2018}, the transitions
temperatures for $(3,4)$ QHP and the dual $(3,4)$ QHP models at $p=0$
are outside of this interval.

\section{Conclusion}\label{sec:conclusion}

In conclusion, we have studied error correction properties of the
finite-rate family of quantum hypergraph product codes obtained from
$\mathbb{B}(3,4)$ Gallager ensemble of classical binary codes, by
combining the threshold calculation using a cluster-based decoder
approximating  minimum-energy decoding with the analysis of the
phase diagram of the associated spin models (Fig.~\ref{fig1}).
Rigorous analytical bounds for the decodable region are constructed by
analyzing the properties of the homological difference
(\ref{eq:homological-difference}), equal to the logarithm of the
conditional decoding probability with the negative sign.

The estimated minimum-weight decoding threshold error rate for this
code family is in the vicinity of $7.0\%$.  This estimate is not so far
from the perfect matching algorithm threshold of $10.4\%$ for the
toric codes \cite{Dennis-Kitaev-Landahl-Preskill-2002}, and is much
higher compared to the analytic lower bound of $0.7\%$ obtained in
Ref.~\cite{Dumer-Kovalev-Pryadko-bnd-2015}.

The most striking feature of the phase diagram of the associated spin
models originating from the finite asymptotic rate [$R=1/25$ for
$(3,4)$ QHP codes] is the deviation of the transition lines from the
self-dual temperature at $p=0$.  In fact, the transitions temperatures
of the two dual models deviate from each other throughout the
small-$p$ region we studied.  We expect the multicritical points,
where the corresponding transition lines intersect the Nishimori line,
also to be different, contrary to the implicit assumption in
Ref.~\onlinecite{Kovalev-Pryadko-SG-2015}.

Notice that the horizontal position $p_\mathrm{bnd}$ of the rightmost
point of the region where Theorem \ref{th:lts-convergence-codes}
guarantees decodability with asymptotic probability one
(magenta-shaded region in Fig.~\ref{fig1}) coincides with the analytic
lower bound for the energy-based decoding from
Ref.~\onlinecite{Dumer-Kovalev-Pryadko-bnd-2015}.  While the former
region is entirely located above the Nishimori line, the
minimum-energy decoding threshold corresponds to $T=0$.  A point on
the Nishimori line correspond to maximum-likelihood decoding at the
corresponding $p$.  This guarantees that the portion of the Nishimori
line for $p\le p_\mathrm{bnd}$ is also inside the decodable region.
It is reasonable to expect that for $p\le p_\mathrm{bnd}$, the entire
interval of temperatures below the bound of Theorem
\ref{th:lts-convergence-codes} would be in the decodable region.
However, construction of the corresponding analytical bound is still an
open problem. 

\begin{acknowledgments}
This work was supported in part by the NSF under Grants No.\
PHY-1415600 (AAK) and  PHY-1416578 (LPP). The computations were
performed utilizing the Holland Computing Center of the University of
Nebraska. 
\end{acknowledgments}

\appendix


\section{Proof of Eqs.~(\ref{eq:homo-generic-bounds}) to
  (\ref{eq:homo-duality}).}
\label{app:homo-diff}

\noindent(\textbf{i}) The lower bound in
Eq.~(\ref{eq:homo-generic-bounds}),
$$
0\le \Delta F_\mathbf{0}(G,H;K)\le k\ln2 ,\eqno(\ref{eq:homo-generic-bounds})
$$
is trivial to prove, since $Z_\mathbf{0}(H^*;K)$ is a sum of positive
terms which include every term present in $Z_\mathbf{0}(G;K)$.  To
prove the upper bound, notice that for any
$\mathbf{e}\in\mathbb{F}_2^n$,
${Z}_\mathbf{e}(G;K)\le {Z}_\mathbf{0}(G;K)$; this can be proved by
comparing the corresponding expansions in powers of $\tanh K$.  The
expression for $Z_\mathbf{0}(H^*;K)=\sum_\mathbf{c}Z_\mathbf{c}(G;K)$
includes the summation over $2^{k}$ distinct defect vectors
$\mathbf{c}$, thus $Z_\mathbf{0}(H^*;K)\le 2^k Z_\mathbf{0}(G;K)$,
which gives the upper bound in
Eq.~(\ref{eq:homo-generic-bounds}).\smallskip

\noindent(\textbf{ii}) 
The inequality 
$$
\Delta F_\mathbf{e}(G,H;K)- \Delta F_\mathbf{0}(G,H;K)\ge 0
\eqno(\ref{eq:homo-diff-disorder})
$$
is derived with the help of the duality~(\ref{eq:em-duality}) which
maps the l.h.s.\ into the difference of the logarithms of the
averages,
\begin{eqnarray*}
\Delta F_\mathbf{e}-\Delta F_\mathbf{0}  
  &= &  
       \ln {Z_\mathbf{e}(H^*;K)\over Z_\mathbf{0}(H^*;K)}-\ln
       {Z_\mathbf{e}(G;K)\over Z_\mathbf{0}(G;K)} \\
  &=& 
        \ln \bigl\langle R^\mathbf{e}\bigr\rangle_{H;K^*}-        
        \ln \bigl\langle R^\mathbf{e}\bigr\rangle_{G^*;K^*};
\end{eqnarray*}
the difference is non-negative by the GKS second
inequality\cite{Griffiths-1967,Kelly-Sherman-1968} (average in the
first term can be obtained from that on the right by applying an
infinite field at the $k$ additional spins).

\noindent(\textbf{iii}) The duality relation 
$$
\Delta F_\mathbf{0}
(G,H;K)=k\ln 2-\Delta F_\mathbf{0}
(H,G;K^*).  \eqno(\ref{eq:homo-duality})
$$
is a simple consequence of Eq.~(\ref{eq:em-duality}) with
$\mathbf{e}=\mathbf{m}=\mathbf{0}$ and the definition of the dual
matrices $G^*$, $H^*$.  Let $r_G$ and $r_H$ denote the numbers of rows
in $G$ and $H$, respectively.  By construction, the dual matrices
$G^*$ and $H^*$ have $r_{G^*}=r_H+k$ and $r_{H^*}=r_G+k$ rows, and
their ranks are $\rank G^*=n-\rank G$, $\rank H^*=n-\rank H$.  We
have,
\begin{eqnarray*}
{Z_{\bf 0}(H^*,K)\over Z_{\bf 0}(G,K)}
&=&{2^{r_H^*-r_H+\rank H}\over 2^{r_G-r_G^*+\rank G^*}}
{Z_{\bf 0}(H,K^*)\over Z_{\bf 0}(G^*,K^*)}\\
&=&2^{k}
{Z_{\bf 0}(H,K^*)\over Z_{\bf 0}(G^*,K^*)}.
\end{eqnarray*}
Eq.~(\ref{eq:homo-duality}) is obtained by taking the logarithm.

\section{Proof of Lemma \ref{th:codes-F}}
\label{app:lemma1}

\thcodesF*

\begin{proof}
  Part (a) immediately follows from the convexity of the exponential function, 
$$
\left[P(\mathbf{e}|\mathbf{e}H^T)\right]_p\ge \exp\left[\ln
  {Z_\mathbf{e}(G;K)\over
    Z_\mathbf{e}(H^*;K)}\right]_p=e^{-\Delta F_p(G,H;K)}.
$$
Part (b) follows from the trivial bounds on the partition function,
$2^r e^{-Kn}\le Z_\mathbf{e}(G;K)\le 2^r e^{Kn}$, where $G$ is an
$r\times n$ matrix.  This gives a lower bound for the conditional
probability (\ref{eq:succ-decoding}),
\begin{equation}
\ln P({\bf e}|\mathbf{s})\ge \ln\left({2^r e^{-Kn}\over
    2^{r+k}e^{Kn}}\right)={-n(2K+R)\ln2}. \label{eq:decod-prob-bnd}
\end{equation}
Now, for some $\delta>0$, let us say that a ``good''
disorder configuration $\mathbf{e}$ corresponds to
$P({\bf e}|\mathbf{e}H^T)\ge1-\delta$, to obtain 
\begin{eqnarray} \nonumber 
  [\Delta F_\mathbf{e}]_p
  &=&   -[\ln P({\bf e}|\mathbf{e}H^T)]_p\\
  \nonumber 
  &\le& {nM}\,P_\mathrm{bad}+(1-P_\mathrm{bad})\ln{1\over  1-\delta}\\ 
  &\le&  nM \,P_\mathrm{bad}+\ln{1\over 1-\delta},
        \label{eq:deltaF-bnd}
\end{eqnarray}
where
$M=(2K+R)\ln2$ is the constant in the r.h.s.\ of
Eq.~(\ref{eq:decod-prob-bnd}), and
$P_\mathrm{bad}=1-P_\mathrm{good}$ is the net probability to encounter
a bad configuration.  A similar chain of inequalities gives an upper
bound for $P_\mathrm{bad}$:
\begin{eqnarray*}
P_\mathrm{succ}&=&[P(\mathbf{e}|H^T\mathbf{e})]_p\\
          &\le&P_\mathrm{good}+(1-P_\mathrm{good})(1-\delta)\\
          &=&1-(1-P_\mathrm{good})\delta;\text{\ thus\ }\\
1-P_\mathrm{succ}&\ge&(1-P_\mathrm{good})\delta=P_\mathrm{bad}\, \delta,
\end{eqnarray*}
Combining with Eq.~(\ref{eq:deltaF-bnd}), this gives for the success
probability~(\ref{eq:psucc}), at a fixed $0<\delta<1$:
\begin{eqnarray}
  \nonumber 
1-P_\mathrm{succ}&\ge& P_\mathrm{bad}\,\delta \\
  \nonumber 
  &\ge& \delta\, {[\Delta
  F_\mathbf{e}]_p+\ln (1-\delta)\over n M}\\
  &\stackrel{n\to\infty}{=}&\delta\,{[\Delta
  f_\mathbf{e}]_p\over (2K+R)\ln2}>0,\label{eq:1}  
\end{eqnarray}
which limits $P_\mathrm{succ}$ from above, away from one.

\end{proof}

\section{Proof of Theorem \ref{th:lts-convergence-codes}}
\label{sec:lts-convergence-codes}

\ltsconvergencecodes*

The statement of the theorem immediately follows from the positivity
of $\Delta F_\mathbf{e}(G,H;K)$, see
Eq.~(\ref{eq:homo-generic-bounds}), and the following Lemma:
\begin{lemma}
  \label{th:homological-bound}  
  Consider a pair of Ising models defined in terms of matrices $G$ and
  $H$ with orthogonal rows, such that the matrix $H$ has a maximum row
  weight $m$.  Let $d_G$ denote the CSS distance (\ref{eq:CSS-dist}),
  the minimum weight of a defect
  $\mathbf{c}\in{\cal C}_H^\perp\setminus {\cal C}_G$.  Denote
  $C\equiv e^{-2K}(1-p)+e^{2K}p$, and assume that 
  $ (m-1) C<1$.  Then, the disorder-averaged homological difference
  (\ref{eq:homological-difference}) satisfies
    \begin{equation}
[\Delta F(G,H;K)]_p\le n \,{(m-1)^{d_G} C^{d_G+1}\over
    1-(m-1)C} .\label{eq:homological-bound}
\end{equation}
\end{lemma}

\begin{proof}
It is convenient to represent the partition function (\ref{eq:Z}) 
in the form 
$$Z_{\bf e}(G;K)=
e^{ K
  n}\sum_{\boldsymbol{\varepsilon}\simeq\mathbf{0}}e^{-2K\wgt(\mathbf{e}+\boldsymbol{\varepsilon})},$$
where the notation $\boldsymbol{\varepsilon}\simeq \mathbf{0}$
indicates that $\boldsymbol{\varepsilon}$ is in the trivial degeneracy
class, that is, it can be represented as a linear combination of rows
of $G$, $\boldsymbol{\varepsilon}=\boldsymbol\alpha G$, and
$\wgt(\mathbf{e}+\boldsymbol{\varepsilon})$ is the total number of
flipped bonds with the spins $S_i=(-1)^{\alpha_i}$.  In comparison,
$$
Z_\mathbf{e}(H^*;K)=e^{ K
  n}\sum_{\boldsymbol\varepsilon:
  H\boldsymbol\varepsilon^T=\mathbf{0}}e^{-2K\wgt(\mathbf{e}+\boldsymbol{\varepsilon})};
$$
here the summation is over all vectors
$\boldsymbol{\varepsilon}\in \mathbb{F}_2^n$ which are orthogonal to
the rows of $H$.  Let us consider a decomposition of any such binary
vector $\boldsymbol{\varepsilon}$ into \emph{irreducible}
components\cite{Dumer-Kovalev-Pryadko-bnd-2015},
$\boldsymbol{\varepsilon}=\boldsymbol{\varepsilon}_1
+\boldsymbol{\varepsilon}_2+\ldots$, where supports of different
vectors in the decomposition do not overlap,
$\boldsymbol{\varepsilon}_i\cap \boldsymbol{\varepsilon}_j=\emptyset$
if $i\neq j$.  The requirement is that each component
$\boldsymbol{\varepsilon}_i$ be orthogonal to the rows of $H$,
and cannot be further decomposed into a sum of non-overlapping
zero-syndrome vectors (such a decomposition is not necessarily
unique).  Now, group all of the components which are trivial,
$\boldsymbol{\varepsilon}_i \simeq \mathbf{0}$, into the vector
$\boldsymbol{\varepsilon}''$, and the non-trivial components into the
vector $\boldsymbol{\varepsilon}'$, so that
$\boldsymbol{\varepsilon}=\boldsymbol{\varepsilon}'
+\boldsymbol{\varepsilon}''$, where
$\boldsymbol{\varepsilon}'\cap \boldsymbol{\varepsilon}''=\emptyset$,
vector $\boldsymbol{\varepsilon}'$ is a sum of non-trivial
non-overlapping codewords
$\mathbf{c}_j\in \mathcal{C}_H^\perp\setminus \mathcal{C}_G$, and the
remainder is trivial, $\boldsymbol{\varepsilon}''\simeq\mathbf{0}$.

Given such a decomposition for each vector
$\boldsymbol{\varepsilon}\in\mathcal{C}_H^\perp$, we can construct an
upper bound for the ratio,
\begin{eqnarray*}
  {Z_\mathbf{e}(H^*;K)\over Z_\mathbf{e}(G;K)}
  &=&
      {\displaystyle \sum_{\boldsymbol{\varepsilon}:
      H\boldsymbol{\varepsilon}^T=\mathbf{0}}e^{-2K\wgt(\mathbf{e}+\boldsymbol{\varepsilon})}\over
      \displaystyle
      \sum_{\boldsymbol{\varepsilon}\simeq\mathbf{0}}
      e^{-2K\wgt(\mathbf{e}+\boldsymbol{\varepsilon})}}\\ 
  &\le& 
        \sum_{\boldsymbol{\varepsilon}'}{\displaystyle 
        \sum_{\boldsymbol{\varepsilon}''\simeq 0:\boldsymbol{\varepsilon}''\cap  \boldsymbol{\varepsilon}'=\emptyset} e^{-2K
    \wgt(\mathbf{e}+\boldsymbol{\varepsilon}'+\boldsymbol{\varepsilon}'')}\over \displaystyle
\sum_{\boldsymbol{\varepsilon}''\simeq 0:\boldsymbol{\varepsilon}''\cap
    \boldsymbol{\varepsilon}'=\emptyset} e^{-2K \wgt(\mathbf{e}+\boldsymbol{\varepsilon}'')}},
\end{eqnarray*}
where the outside summation is over $\boldsymbol{\varepsilon}'$, a sum of
non-overlapping irreducible codewords, and (for a given
$\boldsymbol{\varepsilon}'$) we reduced the denominator by dropping the terms which
overlap with $\boldsymbol{\varepsilon}'$, to match the corresponding sum in the
numerator.  The ratios for each $\boldsymbol{\varepsilon}'$ can now be trivially
calculated in terms of the weight of $\mathbf{e}$ in the support of
$\boldsymbol{\varepsilon}'$, which we denote as $\wgt(\mathbf{e}')$.  We have
$$
  {Z_\mathbf{e}(H^*;K)\over Z_\mathbf{e}(G;K)}\le
  \sum_{\boldsymbol{\varepsilon}'} e^{-2K [\wgt(\boldsymbol{\varepsilon}')-2\wgt(\mathbf{e}')]},
$$
and the corresponding average 
$$
\left[  {Z_\mathbf{e}(H^*;K)\over Z_\mathbf{e}(G;K)}\right]_p\le
\sum_{\boldsymbol{\varepsilon}'} C^{\wgt(\boldsymbol{\varepsilon}')}, 
$$
where the constant $C\equiv (1-p) e^{-2K}+p e^{2K}$.  The summation is
over sums of irreducible non-overlapping codewords,
$\boldsymbol{\varepsilon}'=\mathbf{c}_1+\mathbf{c}_2+\ldots+
\mathbf{c}_m$; we can further increase the r.h.s.\ if we allow the
overlaps between the codewords, to obtain
$$
\left[  {Z_\mathbf{e}(H^*;K)\over Z_\mathbf{e}(G;K)}\right]_p\le
\exp\Bigl(
  \sum_{\mathbf{c}} C^{\wgt(\mathbf{c})}
\Bigr), 
$$
where the summation is now done over irreducible codewords. 
The bound for $[\Delta F_\mathbf{e}(G,H;K)]_p$ is
obtained using the concavity of the logarithm,
$$
\left[
\ln  {Z_\mathbf{e}(H^*;K)\over Z_\mathbf{e}(G;K)}\right]_p\le 
\ln \left[
 {Z_\mathbf{e}(H^*;K)\over Z_\mathbf{e}(G;K)}\right]_p\le
\sum_\mathbf{c} C^{\wgt(\mathbf{c})}.
$$
The final step is to bound the number of
irreducible codewords by the number of the vectors orthogonal to the
rows of $H$ of weight $d_G$ or larger.  For the number $N_w$ of vectors in
$\mathcal{C}_H^\perp$ of weight $w$ one
has\cite{Dumer-Kovalev-Pryadko-2014,Dumer-Kovalev-Pryadko-bnd-2015}
$N_w\le n (m-1)^w$; summation over $w\ge d_G$ gives
Eq.~(\ref{eq:homological-bound}). 
\end{proof}

\section{Proof of Theorem \ref{th:upper-homo-bound}}
\label{app:upper-decod}

\upperhomobound*

\begin{proof}
  By Eq.~(\ref{eq:homo-diff-disorder}), to establish the upper bound,
  we can work at $p=0$.  Let $T_1=1/K_1$ and
  $T_2=1/K_2$ respectively be the upper boundaries of the
  homological regions such that for $\Delta f_\mathbf{0}(G,H;K)=0$ and
  $\Delta f_\mathbf{0}(H,G;K)=0$.  By duality (\ref{eq:homo-duality}),
  $\Delta f_\mathbf{0}(G,H;K_2^*)=R\ln2$.  On the other hand, 
  the   derivative of $f_\mathbf{0}(G;K)$ with respect to $K$ is the
  average energy per bond, 
$$
\partial_K f_\mathbf{0}(G;K)=-n^{-1} \sum_b \langle R_b\rangle_{G;K};
$$
using the GKS inequalities we obtain 
$$
0\le \langle R_b\rangle_{H^*;K}\le \langle R_b\rangle_{G;K}\le 1.
$$
This implies the derivative of $-\Delta f_\mathbf{0}(G,H;K)$ with
respect to $K$ must be in the interval $(0,1)$.  Consequently,
$K_1-K_2^*\ge R\ln2$.  Similar arguments with $G$ and $H$ interchanged
gives $K_2-K_1^*\ge R\ln2$.  If we define $K_{\rm max}=\max(K_1,K_2)$,
then it satisfies Eq.~(\ref{eq:ineq-decod}).
\end{proof}
\bibliography{lpp,spin,sg,qc_all,more_qc,percol,ldpc,linalg,teach}

\end{document}